\documentclass[11pt,oneside,a4paper]{article}
\usepackage{fullpage}
\usepackage[affil-it]{authblk}
\usepackage{amsmath}
\usepackage{amsthm}
\usepackage{amsfonts}
\usepackage{amssymb}
\usepackage{mathtools}
\usepackage{algpseudocode}
\usepackage{algorithm}
\usepackage{graphicx}
\usepackage{subcaption}
\usepackage{units}

\newtheorem{theorem}{Theorem}
\newtheorem{lemma}{Lemma}
\newtheorem{proposition}{Proposition}
\newtheorem{corollary}{Corollary}

\theoremstyle{remark}

\newcommand{\ie}{i.\,e., }

\newcommand{\PSD}{\mathbb{S}^m_{+}}
\newcommand{\CP}{\mathcal{CP}_m}
\newcommand{\DNN}{\mathcal{DNN}_m}
\newcommand{\one}{1}
\newcommand{\Sr}{S^{-}}
\newcommand{\Sc}{S^{|}}


\title{\textbf{Revisiting the Graph Isomorphism Problem with Semidefinite Programming}}
\author{Giannis Nikolentzos\textsuperscript{\rm 1} and Michalis Vazirgiannis\textsuperscript{\rm 2}\\
\textnormal{\textsuperscript{\rm 1}University of Peloponnese, Greece}\\
\textnormal{\textsuperscript{\rm 2}\'Ecole Polytechnique, France}\\
}

\date{preliminary version\footnote{Note: this article has not been peer reviewed yet.} \\November 2026}

\begin{document}
\maketitle

\begin{abstract}
\noindent It is well-known that two graphs $G_1$ and $G_2$ on $n$ vertices are isomorphic if and only if their compatibility graph, an auxiliary graph on $n^2$ vertices, contains a clique of order $n$.
In this paper, we give an exact conic formulation of this clique problem.
We show that $G_1$ and $G_2$ are isomorphic if and only if a system of linear constraints over the cone of completely positive matrices is feasible.
Since deciding membership in the completely positive cone is NP-hard, this formulation does not by itself yield an efficient algorithm.
Relaxing complete positivity to positive semidefiniteness yields a semidefinite program whose infeasibility certifies that the two graphs are not isomorphic.
We establish basic properties of this relaxation.
Specifically, it is at least as strong as the necessary condition given by the Lov\'asz $\vartheta$ number of the complement of the compatibility graph, it never admits a strictly feasible solution, and it is implied by the degree-two level of the Lasserre hierarchy applied to a natural integer formulation of graph isomorphism.
In light of known lower bounds for the Lasserre hierarchy, the latter indicates that the relaxation cannot decide isomorphism for all pairs of graphs.
We leave open the question of characterizing the pairs of graphs for which the relaxation is tight.
\end{abstract}

\section{Introduction}
Graph isomorphism is one of those few fundamental problems in NP whose computational status still remains unknown \cite{garey1979computers}.
Roughly speaking, the graph isomorphism problem asks whether two graphs are structurally identical or not.
The problem is clearly in NP.
However, it has been neither proven NP-complete nor found to be solved by a polynomial time algorithm.
In fact, there is strong evidence that graph isomorphism is not NP-complete since it has been shown that the problem is located in the low hierarchy of NP \cite{schoning1988graph}.
This implies that if the problem was NP-complete, then the polynomial time hierarchy would collapse to its second level.

Over the years, algorithms of different nature have been developed to attack the problem.
Traditionally, those that draw ideas from group theory turn out to be the most promising.
One of these group-theoretic algorithms was proposed by Babai and Luks in 1983 \cite{babai1983canonical}.
The algorithm combines a preprocessing procedure proposed by Zemlyachenko et al. \cite{zemlyachenko1985graph} with an efficient algorithm for solving graph isomorphism on graphs of bounded degree \cite{luks1982isomorphism}.
Its computational complexity is $2^{\mathcal{O}(\sqrt{n \log n})}$ where $n$ denotes the number of vertices.
Despite decades of active research, no progress had been achieved, and this was the best known algorithm until recently when Babai presented an algorithm that solves the graph isomorphism problem in quasi-polynomial time \cite{babai2016graph}.
It should be mentioned that while the complexity status of the graph isomorphism for general graphs remains a mystery, for many restricted graph classes, polynomial time algorithms are known.
This is, for example, the case for planar graphs \cite{hopcroft1974linear}, graphs of bounded degree \cite{luks1982isomorphism}, or graphs with bounded eigenvalue multiplicity \cite{babai1982isomorphism}.
It should also be noted that there exist several algorithms which have proven very efficient for graphs of practical interest \cite{mckay2014practical,junttila2007engineering,darga2008faster}.
Interestingly, these algorithms are very different from the ones that offer the lowest worst-case complexities.
This indicates that there is a wide gap between theory and practice.

A common heuristic for testing two graphs for isomorphism is the color refinement algorithm \cite{weisfeiler1968reduction}.
Given an initial coloring of the vertices of two graphs, the algorithm iteratively assigns new colors to the vertices.
The new color of a vertex is a combination of its current color and the multiset of colors of its neighbors.
The refinement step can be repeated until the coloring converges.
If, after some iteration, the two graphs have a different number of vertices of a certain color, then the graphs are not isomorphic to each other, and the algorithm terminates.
However, convergence of the algorithm does not imply an isomorphism between the two graphs.
For instance, the algorithm fails when applied to pairs of non-isomorphic regular graphs of the same degree and order.
Nonetheless, it is a powerful heuristic, which can successfully test isomorphism for many classes of graphs \cite{babai1979canonical}.
The $k$-dimensional Weisfeiler-Leman algorithm ($k$-WL) generalizes the algorithm described above by iteratively refining colorings of $k$-tuples of vertices instead of single vertices.
The complexity of the $k$-WL algorithm is $n^{O(k)}$.
There was hope that $k$-WL for small values of $k$ might solve graph isomorphism.
However, Cai, F\"urer, and Immerman constructed families of pairs of non-isomorphic graphs on $n$ vertices that are not distinguished by the $k$-WL algorithm even when $k$ grows linearly with $n$, \ie for all $k \leq cn$ for some constant $c > 0$ \cite{cai1992optimal}.
It should be mentioned that simple group-theoretic algorithms can successfully decide if the Cai-F\"urer-Immerman graphs are isomorphic to each other in polynomial time.

Besides the above algorithms, there are also several scalable heuristics for graph isomorphism which are based on continuous optimization.
In these heuristics, the discrete search problem in the space of permutation matrices is replaced by an optimization problem with continuous variables, enabling the use of efficient continuous optimization algorithms.
Formally, for any two graphs on $n$ vertices with respective $n \times n$ adjacency matrices $A_1$ and $A_2$, the optimization problem consists in minimizing the function $\|A_1 - P A_2 P^\top\|_F$ over all $P \in \Pi$, where $\Pi$ denotes the set of $n \times n$ permutation matrices, and $\|\cdot\|_F$ is the Frobenius matrix norm \cite{aflalo2015convex}.
Therefore, the problem of graph isomorphism can be reformulated as the problem of minimizing the above function over the set of permutation matrices.
The two graphs are isomorphic to each other if there exists a permutation matrix $P$ for which the above function is equal to $0$.
Note also that other objectives have also been proposed in the literature, this being perhaps the most common.
This problem has a combinatorial nature and there is no known polynomial algorithm to solve it.
Numerous approximate methods have been developed.
Most of these methods replace the space of permutations by the space of doubly-stochastic matrices.
Let $\mathcal{D}$ denote the set of $n \times n$ doubly stochastic matrices, \ie nonnegative matrices with row and column sums each equal to $1$.
Since $\|A_1 - P A_2 P^\top\|_F = \|A_1 P - P A_2\|_F$ for every permutation matrix $P$, a convex relaxation is obtained by minimizing the function $\| A_1 D - D A_2 \|_F^2$ over all $D \in \mathcal{D}$ \cite{aflalo2015convex}.
There is a polynomial-time algorithm for exactly solving this convex quadratic program \cite{goldfarb1990n}.
However, due to relaxation, even if there exists a doubly stochastic matrix $D$ for which the objective function is equal to $0$, there is no guarantee that the two graphs are isomorphic to each other.

Recently, there is a growing interest in investigating the power of linear/semidefinite programming hierarchies in the context of graph isomorphism.
Specifically, it has been shown that the $k^{th}$ level of the Sherali-Adams linear programming hierarchy is sandwiched between the $(k-1)$-WL and the $k$-WL algorithms \cite{atserias2013sherali}.
More importantly, it has also been shown that there are pairs of non-isomorphic graphs of order $n$ for which the $\Omega(n)^{th}$ level of the Lasserre hierarchy cannot prove that they are not isomorphic to each other \cite{o2014hardness,snook2014graph}.
These graphs correspond to constructions similar to the Cai-F\"urer-Immerman graphs, and emerge from instances of the 3XOR problem.
In these hierarchies, non-isomorphism is certified by infeasibility of the relaxation.
The above results thus show that no polynomial-size level of these hierarchies decides graph isomorphism in general.

In this paper, we study graph isomorphism through a conic lens.
The main tools employed are a \textit{compatibility graph}, \ie an auxiliary graph structure that is useful for solving general graph and subgraph isomorphism problems, and \textit{completely positive} and \textit{semidefinite} programming.
Given two graphs of order $n$, we build their compatibility graph of order $n^2$.
Testing the two graphs for isomorphism is equivalent to determining whether the compatibility graph contains a clique of order $n$.
Our contributions are as follows:
(i) we give a system of linear constraints such that the two graphs are isomorphic if and only if the system has a completely positive solution (Theorem~\ref{thm:cp_exact} and Corollary~\ref{cor:gi_cp});
(ii) we study the semidefinite relaxation obtained by replacing complete positivity with positive semidefiniteness, whose infeasibility certifies non-isomorphism; we show that it is at least as strong as the test based on the Lov\'asz $\vartheta$ number of the complement of the compatibility graph (Proposition~\ref{prop:theta}), that it never has a strictly feasible solution (Proposition~\ref{prop:kernel}), and that it is implied by the degree-two level of the Lasserre hierarchy for a natural integer formulation of graph isomorphism (Proposition~\ref{prop:lasserre});
and (iii) we formulate the question of when this relaxation is tight.
Since deciding membership in the completely positive cone is NP-hard \cite{dickinson2014computational}, contribution (i) does not yield an efficient algorithm, and by (ii) and the lower bounds of \cite{o2014hardness,snook2014graph}, the semidefinite relaxation is not expected to decide isomorphism for all pairs of graphs.
It should be mentioned that our work is not the first to apply continuous optimization approaches to the problem of determining whether the compatibility graph contains a clique of order $n$.
In a previous study, Pelillo developed a heuristic for computing the clique number of the compatibility graph \cite{pelillo1999replicator}.
That approach is a local method for a nonconvex program, and may get stuck on some local optimum of the objective function.
Here, instead, we consider convex (conic) formulations.

\section{Preliminaries}
In this Section, we first define our notation, and we then introduce the concept of a compatibility graph.
We show that the graph isomorphism problem is equivalent to finding if the compatibility graph contains a clique of specific size.
We also present a necessary condition for isomorphism which is based on the Lov\'asz $\vartheta$ number.

\subsection{Graph Theory Notation and Terminology}
Let $G = (V,E)$ be an undirected and unweighted graph consisting of a set $V$ of vertices and a set $E$ of edges between them.
We will denote by $n$ the number of vertices.
The adjacency matrix of $G$ is a symmetric matrix $A \in \mathbb{R}^{n \times n}$ defined as follows: $A_{i,j} = 1$ if $(i,j) \in E$, and $0$ otherwise.
Note that since the graph is undirected, $(i,j) \in E$ if and only if $(j,i) \in E$, for all $i,j \in V$.
Two graphs $G_1 = (V_1,E_1)$ and $G_2 = (V_2,E_2)$ are isomorphic (denoted by $G_1 \cong G_2$), if there is a bijective mapping $\phi : V_1 \rightarrow V_2$ such that $(i,j) \in E_1$ if and only if $(\phi(i), \phi(j)) \in E_2$.
We denote by $\omega(G)$ and $\chi(G)$ the clique number and the chromatic number of $G$, respectively, and by $\bar{G}$ the complement of $G$.

We next present the notion of a compatibility graph, \ie a structure that is very useful for solving graph/subgraph isomorphism and maximum common subgraph problems.
These auxiliary graph structures have been proposed independently by several authors \cite{levi1973note,barrow1976subgraph,kozen1978clique}, while they also lie at the core of several algorithms \cite{pelillo1999replicator,koch2001enumerating,kriege2012subgraph}.
Furthermore, different authors have used different names to describe them.
For instance, compatibility graphs, association graphs, derived graphs, M-graphs, and product graphs are all names that have been coined to describe these structures.
In what follows, we will use the name \textit{compatibility graph} to refer to them.
Formally, given two graphs $G_1 = (V_1 , E_1)$ and $G_2 = (V_2 , E_2)$, their compatibility graph $G_c = (V_c , E_c)$ is a graph with vertex set $V_c = V_1 \times V_2$.
An edge is drawn between two vertices $(v_1, u_1), (v_2, u_2) \in V_c$ if and only if $v_1 \neq v_2$, $u_1 \neq u_2$ and either $e_1 = (v_1, v_2) \in E_1$ and $e_2 = (u_1, u_2) \in E_2$ or $e_1 \not \in E_1$ and $e_2 \not \in E_2$.
Clearly, there are two types of edges in a compatibility graph: (1) edges that represent common adjacency, and (2) edges that represent common non-adjacency.
An example of the compatibility graph that emerges from two $P_3$ graphs is illustrated in Figure~\ref{fig:compatibility_graph}.
\begin{figure*}[t]
    \centering
    \includegraphics[width=.42\linewidth]{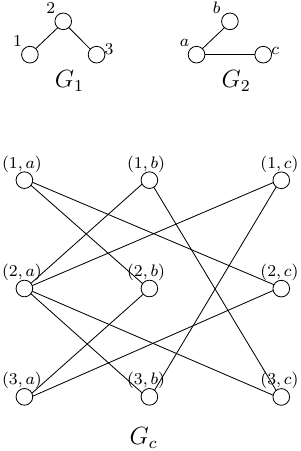}
    \caption{Two graphs (top left and right) and their compatibility graph (bottom).}
    \label{fig:compatibility_graph}
\end{figure*}

Levi established a relation between isomorphic subgraphs of two graphs and cliques in their compatibility graph \cite{levi1973note}.
Specifically, if some vertices $(v_1, u_1), (v_2, u_2), \ldots, (v_k, u_k) \in V_c$ form a clique, and are thus pairwise adjacent, then the subgraph in $G_1$ induced by the vertices $v_1, v_2, \ldots, v_k$ is isomorphic to the subgraph in $G_2$ induced by the vertices $u_1, u_2, \ldots, u_k$.
The isomorphism is given by the vertices $(v_1, u_1), (v_2, u_2), \ldots, (v_k, u_k) \in V_c$ of the compatibility graph that form the clique, \ie $\phi(v_1) = u_1, \phi(v_2) = u_2, \ldots, \phi(v_k) = u_k$.

The next Theorem establishes an equivalence between the graph isomorphism problem and the maximum clique problem on the compatibility graph, and corresponds to a sub-instance of Levi's result (see also \cite{kozen1978clique}).
For completeness, we also provide the proof.
Note that if $G_1$ and $G_2$ are two graphs of order $n$, the maximum clique of their compatibility graph will consist of at most $n$ vertices, since two vertices of $G_c$ that share their first coordinate are never adjacent.
\begin{theorem}
  Let $G_1$ and $G_2$ be two graphs of order $n$, and let $G_c$ be their compatibility graph.
  Then, $G_1$ and $G_2$ are isomorphic if and only if $G_c$ contains a clique of order $n$, \ie $\omega(G_c) = n$.
  \label{thm:levi}
\end{theorem}
\begin{proof}
  If $G_1 = (V_1,E_1)$ and $G_2 = (V_2,E_2)$ are isomorphic, then there exists a bijective mapping $\phi : V_1 \rightarrow V_2$ such that $(v_i,v_j) \in E_1$ if and only if $(\phi(v_i),\phi(v_j)) \in E_2$.
  Consider the $n$ vertices $(v_1, \phi(v_1)), \ldots, (v_n, \phi(v_n)) \in V_c$.
  For every $i \neq j$, we have $v_i \neq v_j$ and, since $\phi$ is injective, $\phi(v_i) \neq \phi(v_j)$; moreover, $(v_i,v_j) \in E_1$ if and only if $(\phi(v_i),\phi(v_j)) \in E_2$.
  Hence, by construction, $\big((v_i, \phi(v_i)), (v_j, \phi(v_j))\big) \in E_c$ for all $i \neq j$.
  These vertices form a clique of order $n$, and since $\omega(G_c) \leq n$, we have $\omega(G_c) = n$.

  For the second part, given a compatibility graph $G_c$ that contains a clique of order $n$, let $(v_1, u_1), \ldots, (v_n, u_n) \in V_c$ denote the $n$ vertices that form the clique.
  Two vertices of $G_c$ can only be adjacent if they differ in both coordinates.
  Hence, $v_1,\ldots,v_n$ are pairwise distinct, and so are $u_1,\ldots,u_n$, \ie $V_1=\{ v_1,\ldots,v_n \}$ and $V_2=\{ u_1,\ldots,u_n \}$.
  Then, the mapping $\phi : V_1 \rightarrow V_2$ with $\phi(v_i) = u_i$ is bijective.
  Since $\big((v_i, \phi(v_i)), (v_j, \phi(v_j))\big) \in E_c$ for all $i,j \in \{ 1,\ldots,n \}$ with $i \neq j$, we have that $(v_i,v_j) \in E_1$ if and only if $(\phi(v_i),\phi(v_j)) \in E_2$.
  Therefore, $G_1$ and $G_2$ are isomorphic to each other.
\end{proof}

Note that the number of isomorphisms between two graphs can be exponential in the number of vertices $n$ of the graphs.
Specifically, if $G_1$ and $G_2$ are isomorphic, then the number of automorphisms of $G_1$ (or of $G_2$) is equal to the number of isomorphisms from $G_1$ to $G_2$.
Hence, if, for instance, $G_1$ and $G_2$ are complete graphs on $n$ vertices, \ie both correspond to the complete graph $K_n$, then the number of isomorphisms between the two graphs is equal to $n!$ since $|Aut(G_1)|=|Aut(G_2)|=|Aut(K_n)|=n!$.
Each isomorphism $\phi : V_1 \rightarrow V_2$ corresponds to an $n$-clique in the compatibility graph, and vice versa.
Therefore, the compatibility graph can contain up to $n!$ cliques of order $n$.
As an example, consider the graphs shown in Figure~\ref{fig:compatibility_graph}.
There are two isomorphisms between $G_1$ and $G_2$.
Hence, their compatibility graph contains exactly two cliques of order $3$, namely $\{(1,b),(2,a),(3,c)\}$ and $\{(1,c),(2,a),(3,b)\}$.

We next define two different groupings of the vertices of the compatibility graph into $n$ parts.
Note that each vertex of the compatibility graph corresponds to a pair of vertices $(v_i, u_j)$ where $v_i \in V_1$ and $u_j \in V_2$.
Let $\Sr = \{S_1^-, S_2^-, \ldots, S_n^-\}$ where $S_i^-$ is a set defined as $S_i^- = \{(v_i, u_1), (v_i, u_2), \ldots, (v_i, u_n) \}$ and $|S_i^-|=n$.
Clearly, the sets $S_1^-, S_2^-, \ldots, S_n^-$ are pairwise disjoint and we have that $S_1^- \cup S_2^- \cup \ldots \cup S_n^- = V_c$.
Therefore, the above family of sets is a partition of $V_c$.
Likewise, let $\Sc = \{S_1^|, S_2^|, \ldots, S_n^|\}$ where $S_i^|$ is a set defined as $S_i^| = \{(v_1, u_i), (v_2, u_i), \ldots, (v_n, u_i) \}$ and $|S_i^||=n$.
Again, the sets $S_1^|, S_2^|, \ldots, S_n^|$ are pairwise disjoint and we have that $S_1^| \cup S_2^| \cup \ldots \cup S_n^| = V_c$.
Therefore, this family of sets is also a partition of $V_c$.
The two partitions of the set of vertices $V_c$ of the compatibility graph of Figure~\ref{fig:compatibility_graph} are illustrated in Figure~\ref{fig:compatibility_graph_partitions}.
For a set $A \subseteq V_c$, we denote by $\one_A \in \{0,1\}^{V_c}$ its indicator vector.

\begin{figure}[t]
    \begin{subfigure}[c]{0.5\linewidth}
    \includegraphics[width=.87\linewidth]{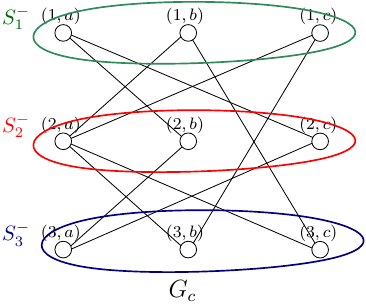}
    \end{subfigure}%
    \begin{subfigure}[c]{0.5\linewidth}
    \includegraphics[width=.8\linewidth]{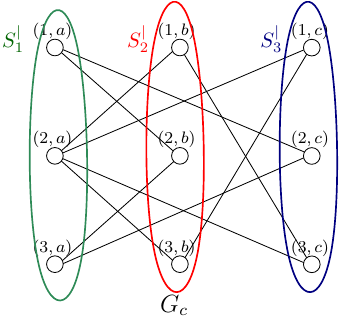}
    \end{subfigure}
    \caption{The two partitions of the vertex set of the compatibility graph of Figure~\ref{fig:compatibility_graph}. Left: the partition $\Sr$, whose part $S_i^-$ contains the vertices whose first coordinate is vertex $i$ of $G_1$. Right: the partition $\Sc$, whose part $S_j^|$ contains the vertices whose second coordinate is the $j^{th}$ vertex of $G_2$.}
    \label{fig:compatibility_graph_partitions}
\end{figure}

We now state two Lemmas which we shall need later on.
\begin{lemma}
  Let $S_1^-, S_2^-, \ldots, S_n^-$ $(S_1^|, S_2^|, \ldots, S_n^|)$ be a partition of the set $V_c$ of vertices of $G_c$ into pairwise disjoint sets each of size $n$ as described above.
  Then, two vertices that belong to the same set are not linked by an edge.
  \label{lemma:lem1}
\end{lemma}
\begin{proof}
  A set $S_i^-$ is defined as $S_i^- = \{(v_i, u_1), (v_i, u_2), \ldots, (v_i, u_n) \}$.
  Any two distinct vertices of $S_i^-$ have the same first coordinate $v_i$, so by the definition of $E_c$, we have $\big( (v_i, u_j), (v_i, u_k) \big) \not \in E_c$ for all $j,k \in [n]$.
  Therefore, no two vertices of the same part are linked by an edge.
  The same argument, applied to the second coordinate, holds for the sets $S_1^|, S_2^|, \ldots, S_n^|$.
\end{proof}

\begin{lemma}
  Let $G_1$ and $G_2$ be two graphs of order $n$, and let $G_c$ be their compatibility graph.
  Let also $C$ be a set of vertices that form a clique of order $n$ in $G_c$.
  Let $\Sr = \{S_1^-, S_2^-, \ldots, S_n^-\}$ $(\Sc = \{S_1^|, S_2^|, \ldots, S_n^|\})$ be a partition of the set $V_c$ of vertices of $G_c$ into pairwise disjoint sets each of size $n$ as described above.
  Then, the function $f: C \rightarrow \Sr$ $(f: C \rightarrow \Sc)$ that maps a vertex of $C$ to the set $S_i^- \in \Sr$ $(S_i^| \in \Sc)$ in which the vertex is contained is bijective.
  In other words, $|C \cap A| = 1$ for every $A \in \Sr \cup \Sc$.
  \label{lemma:lem2}
\end{lemma}
\begin{proof}
  Since $S_1^-,\ldots,S_n^-$ correspond to a partition of the set of vertices of $G_c$, each vertex in $C$ belongs to exactly one of these sets.
  Furthermore, since $S_1^-,\ldots,S_n^-$ are independent sets (Lemma~\ref{lemma:lem1}), each of these sets contains at most one of the vertices in $C$.
  There are $n$ vertices in total in $C$, and therefore, each one of the sets $S_1^-,\ldots,S_n^-$ contains exactly one of these vertices.
  The same holds for the sets $S_1^|,\ldots,S_n^|$.
\end{proof}

\subsection{A Necessary Condition Based on the Lov\'asz Number}
\label{sec:theta}
It follows from Lemma~\ref{lemma:lem1} that the set of vertices $V_c$ of a compatibility graph can be decomposed into $n$ disjoint sets $S_1^-, S_2^-, \ldots, S_n^-$ such that there is no edge that connects vertices of the same set.
Since a compatibility graph can be decomposed into $n$ disjoint sets such that every edge $e \in E_c$ connects a vertex in $S_i^-$ to one in $S_j^-$ with $i \neq j$, the minimum number of colors required for a proper coloring of $G_c$ is no more than $n$.
Therefore, $\chi(G_c) \leq n$.

For an arbitrary graph $G$, it is well-known that $\omega(G) \leq \chi(G)$.
Therefore, a compatibility graph can contain a clique of order $n$ only if $\chi(G_c) = n$.
Computing the chromatic number of a graph is in general an NP-hard problem.
However, we can compute to arbitrary precision in polynomial time a real number $\vartheta(\bar{G})$ that is ``sandwiched'' between the clique number and the chromatic number of a graph $G$, that is $\omega(G) \leq \vartheta(\bar{G}) \leq \chi(G)$ \cite{lovasz1979shannon,grotschel1981ellipsoid}.
Here, $\vartheta$ denotes the Lov\'asz number, which can be written as the semidefinite program
\begin{equation}
  \vartheta(\bar{G}) = \max \Big\{ \langle J, Y \rangle : \mathrm{trace}(Y) = 1, \; Y_{i,j} = 0 \text{ for all } i \neq j \text{ with } (i,j) \notin E, \; Y \succeq 0 \Big\}
  \label{eq:theta}
\end{equation}
where $J$ is the all-ones matrix.
Consequently, if $G_1 \cong G_2$, then $n = \omega(G_c) \leq \vartheta(\bar{G}_c) \leq \chi(G_c) \leq n$, \ie $\vartheta(\bar{G}_c) = n$.
This gives a necessary condition for isomorphism that can be checked (approximately) in polynomial time.
However, $\vartheta(\bar{G})$ can be strictly larger than $\omega(G)$ in general (for instance, for the $5$-cycle, $\omega(C_5) = 2 < \sqrt{5} = \vartheta(\bar{C}_5)$), so the condition $\vartheta(\bar{G}_c) = n$ is not sufficient.
Whether $\vartheta(\bar{G}_c) = n$ can hold when $G_1$ and $G_2$ are not isomorphic is a question we do not resolve here.
In Section~\ref{sec:sdp}, we show that the semidefinite relaxation studied in this paper is at least as strong as this test (Proposition~\ref{prop:theta}).

\section{An Exact Completely Positive Formulation}
\label{sec:cp}
In this Section, we show that the problem of identifying whether a compatibility graph contains a clique of order $n$ can be expressed exactly as a feasibility problem over the cone of completely positive matrices.

\paragraph{Cones.}
Let $m = n^2$ and let $\mathbb{S}^m$ be the set of all $m \times m$ real symmetric matrices, whose rows and columns we index by the vertices of $G_c$.
We use the following three cones: ($1$) the cone $\PSD$ of positive semidefinite matrices, \ie matrices of the form $X = Z Z^\top$ with $Z \in \mathbb{R}^{m \times d}$ for some $d$ (we write $X \succeq 0$); ($2$) the cone $\CP$ of completely positive matrices, \ie matrices of the form $X = Z Z^\top$ where $Z \in \mathbb{R}^{m \times d}$ is elementwise nonnegative for some $d$ \cite{berman2003completely}; and ($3$) the cone $\DNN = \PSD \cap \mathbb{R}^{m \times m}_{\geq 0}$ of doubly nonnegative matrices.
Clearly, $\CP \subseteq \DNN \subseteq \PSD$.
Writing $X = Z Z^\top$, we denote by $z_i$ the $i^{th}$ row of $Z$, so that $X_{i,j} = \langle z_i, z_j \rangle$, \ie $z_i$ is a vector representation of vertex $v_i \in V_c$.

\paragraph{The constraint system.}
For a cone $\mathcal{K} \subseteq \mathbb{S}^m$, consider the following feasibility problem in the variable $X \in \mathbb{S}^m$, which we denote by $\mathrm{P}(\mathcal{K})$:
\begin{subequations}
  \begin{align}
    \mathrm{trace}(X) &= n, \label{eq:opt2_con1} \\
    X_{i,j} &= 0, \; i \neq j, \; (i,j) \not \in E_c \label{eq:opt2_con2} \\
    X_{i,i} &\leq 1, \; i \in V_c \label{eq:opt2_con3} \\
    \sum_{i \in A} \sum_{j\in B} X_{i,j} &= 1, \; (A,B) \in \Sr \times \Sr \label{eq:opt2_con4} \\
    \sum_{i \in A} \sum_{j\in B} X_{i,j} &= 1, \; (A,B) \in \Sc \times \Sc \label{eq:opt2_con5} \\
    \sum_{i \in A} X_{i,j} &= \sum_{i \in B} X_{i,j}, \; j \in V_c, \; (A,B) \in \Sr \times \Sr \label{eq:opt2_con6} \\
    \sum_{i \in A} X_{i,j} &= \sum_{i \in B} X_{i,j}, \; j \in V_c, \; (A,B) \in \Sc \times \Sc \label{eq:opt2_con7} \\
    X &\in \mathcal{K} \label{eq:opt2_con8}
  \end{align}
  \label{eq:opt2}
\end{subequations}
Note that the pairs $(A,B)$ in \eqref{eq:opt2_con4} and \eqref{eq:opt2_con5} include the pairs with $A = B$.
Problem $\mathrm{P}(\mathcal{K})$ has no objective function.
We are only interested in whether it is feasible.
We next show that $\mathrm{P}(\CP)$ is feasible if and only if $G_c$ contains a clique of order $n$.
We first show how a completely positive feasible solution can be constructed from the cliques of order $n$ of $G_c$.

\begin{theorem}
  Suppose that $G_c$ contains $d \geq 1$ distinct cliques $C_1,\ldots,C_d$ of order $n$.
  For each vertex $i \in V_c$, define $z_i \in \mathbb{R}^d$ by $z_i^k = \frac{1}{\sqrt{d}}$ if $i \in C_k$, and $z_i^k = 0$ otherwise, where $z_i^k$ denotes the $k^{th}$ component of $z_i$.
  Then, the matrix $X$ with $X_{i,j} = \langle z_i, z_j \rangle$ is a feasible solution of $\mathrm{P}(\CP)$.
  \label{thm:feasible}
\end{theorem}
\begin{proof}
  Let $Z \in \mathbb{R}^{m \times d}$ be the matrix with rows $z_i$.
  Since $Z$ is elementwise nonnegative, $X = Z Z^\top \in \CP$, so \eqref{eq:opt2_con8} holds.
  By construction, $X_{i,j} = \frac{1}{d} \, |\{ k : i \in C_k \text{ and } j \in C_k \}|$.

  \textit{Constraints \eqref{eq:opt2_con2} and \eqref{eq:opt2_con3}.}
  We have $X_{i,i} = \frac{1}{d} |\{ k : i \in C_k \}| \leq 1$.
  If $i \neq j$ and $(i,j) \notin E_c$, then no clique contains both $i$ and $j$, and thus $X_{i,j} = 0$.

  \textit{Constraint \eqref{eq:opt2_con1}.}
  Since each clique contains exactly $n$ vertices, we have:
  \begin{equation*}
    \mathrm{trace}(X) = \sum_{i \in V_c} \frac{1}{d} |\{ k : i \in C_k \}| = \frac{1}{d} \sum_{k=1}^d |C_k| = \frac{1}{d} \, d \, n = n
  \end{equation*}

  \textit{Constraints \eqref{eq:opt2_con4}--\eqref{eq:opt2_con7}.}
  For a set $A \in \Sr \cup \Sc$, let $w_A = \sum_{i \in A} z_i$.
  By Lemma~\ref{lemma:lem2}, each clique $C_k$ contains exactly one vertex of $A$, and therefore $w_A^k = \frac{1}{\sqrt{d}}$ for every $k$, \ie $w_A = \frac{1}{\sqrt{d}} \one$, where $\one \in \mathbb{R}^d$ is the all-ones vector.
  In particular, $w_A$ is the same vector for all $A \in \Sr \cup \Sc$.
  Hence, for all $A, B \in \Sr$ (or $A, B \in \Sc$), we have:
  \begin{equation*}
    \sum_{i \in A} \sum_{j \in B} X_{i,j} = \langle w_A, w_B \rangle = d \cdot \frac{1}{d} = 1
  \end{equation*}
  which gives \eqref{eq:opt2_con4} and \eqref{eq:opt2_con5}, and, for every $j \in V_c$, we have:
  \begin{equation*}
    \sum_{i \in A} X_{i,j} = \langle w_A, z_j \rangle = \langle w_B, z_j \rangle = \sum_{i \in B} X_{i,j}
  \end{equation*}
  which gives \eqref{eq:opt2_con6} and \eqref{eq:opt2_con7}.
  Hence, all the constraints are satisfied.
\end{proof}

\begin{figure*}[t]
  \centering
  \begin{minipage}{.5\textwidth}
    \centering
    \includegraphics[width=\textwidth]{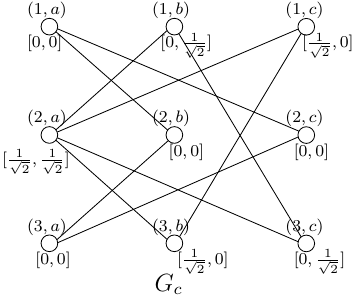}
  \end{minipage}
  \begin{minipage}{\textwidth}
  \centering
  \vspace{.2cm}
  {\footnotesize
    \[
    X =
    \begin{bmatrix}
      0 & 0 & 0 & 0 & 0 & 0 & 0 & 0 & 0 \\
      0 & \nicefrac{1}{2} & 0 & \nicefrac{1}{2} & 0 & 0 & 0 & 0 & \nicefrac{1}{2} \\
      0 & 0 & \nicefrac{1}{2} & \nicefrac{1}{2} & 0 & 0 & 0 & \nicefrac{1}{2} & 0 \\
      0 & \nicefrac{1}{2} & \nicefrac{1}{2} & 1 & 0 & 0 & 0 & \nicefrac{1}{2} & \nicefrac{1}{2} \\
      0 & 0 & 0 & 0 & 0 & 0 & 0 & 0 & 0 \\
      0 & 0 & 0 & 0 & 0 & 0 & 0 & 0 & 0 \\
      0 & 0 & 0 & 0 & 0 & 0 & 0 & 0 & 0 \\
      0 & 0 & \nicefrac{1}{2} & \nicefrac{1}{2} & 0 & 0 & 0 & \nicefrac{1}{2} & 0 \\
      0 & \nicefrac{1}{2} & 0 & \nicefrac{1}{2} & 0 & 0 & 0 & 0 & \nicefrac{1}{2} \\
    \end{bmatrix}
    \]
    }
  \end{minipage}
  \caption{A feasible solution of $\mathrm{P}(\CP)$ for the compatibility graph of Figure~\ref{fig:compatibility_graph}, constructed as in Theorem~\ref{thm:feasible}. There are two cliques of order $3$ in this graph, so each vertex is assigned a $2$-dimensional vector (top). The resulting matrix $X$ (bottom), whose $(i,j)^{th}$ entry is the inner product of the representations of vertices $i$ and $j$; rows and columns are ordered $(1,a),(1,b),(1,c),(2,a),\ldots,(3,c)$.}
  \label{fig:representations}
\end{figure*}

Figure~\ref{fig:representations} illustrates the construction of Theorem~\ref{thm:feasible} for the compatibility graph of Figure~\ref{fig:compatibility_graph}.
There are two cliques of order $3$ in this graph.
Therefore, we assign a $2$-dimensional vector to each of its vertices.

To prove the converse, we need the following Lemma.
\begin{lemma}
  Let $X = Z Z^\top$, where $Z \in \mathbb{R}^{m \times d}$ is elementwise nonnegative with rows $z_i$, $i \in V_c$, and suppose that $X$ satisfies constraint~\eqref{eq:opt2_con2}.
  Let $A \in \Sr \cup \Sc$.
  Then, for each coordinate $r \in \{1,\ldots,d\}$, at most one of the vectors $z_i$, $i \in A$, has a nonzero $r^{th}$ component.
  \label{lemma:lem6}
\end{lemma}
\begin{proof}
  Let $i, j \in A$ with $i \neq j$.
  By Lemma~\ref{lemma:lem1}, $(i,j) \notin E_c$, and therefore, from constraint~\eqref{eq:opt2_con2}, we have:
  \begin{equation*}
    0 = X_{i,j} = \langle z_i, z_j \rangle = z_i^1 \, z_j^1 + z_i^2 \, z_j^2 + \ldots + z_i^d \, z_j^d
  \end{equation*}
  The vectors $z_i,z_j$ are elementwise nonnegative, so every term of the above sum is nonnegative, and hence $z_i^r \, z_j^r = 0$ for all $r=1,\ldots,d$.
  Thus, if $z_i^r > 0$ for some $i \in A$, then $z_j^r = 0$ for all $j \in A \setminus \{ i \}$.
\end{proof}

\begin{theorem}
  The compatibility graph $G_c$ contains a clique of order $n$ if and only if problem $\mathrm{P}(\CP)$ is feasible.
  \label{thm:cp_exact}
\end{theorem}
\begin{proof}
  If $G_c$ contains a clique of order $n$, then $\mathrm{P}(\CP)$ is feasible by Theorem~\ref{thm:feasible}.

  Conversely, let $X$ be a feasible solution of $\mathrm{P}(\CP)$.
  Since $X$ is completely positive, it can be factorized as $X = Z Z^\top$, where $Z \in \mathbb{R}^{m \times d}$ is elementwise nonnegative.
  Let $z_i$ denote the $i^{th}$ row of $Z$, and for each $A \in \Sr$, let $w_A = \sum_{i \in A} z_i$.
  For $A, B \in \Sr$, we have:
  \begin{equation}
    \langle w_A, w_B \rangle = \sum_{i \in A} \sum_{j \in B} \langle z_i, z_j \rangle = \sum_{i \in A} \sum_{j \in B} X_{i,j} = 1
    \label{eq:sum1}
  \end{equation}
  by constraint~\eqref{eq:opt2_con4}.
  Applying \eqref{eq:sum1} with $(A,A)$, $(A,B)$ and $(B,B)$ gives
  $\|w_A - w_B\|^2 = \langle w_A, w_A \rangle - 2 \langle w_A, w_B \rangle + \langle w_B, w_B \rangle = 1 - 2 + 1 = 0$.
  Hence, all the vectors $w_A$, $A \in \Sr$, are equal to a common vector $w$ with $\|w\| = 1$.
  Since $\|w\| = 1$, there is at least one coordinate $r$ with $w^r \neq 0$, and since $Z$ is nonnegative, $w^r > 0$.

  For each $A \in \Sr$, we have $\sum_{i \in A} z_i^r = w^r > 0$, so there exists a vertex $i \in A$ with $z_i^r > 0$ (by Lemma~\ref{lemma:lem6}, this vertex is unique).
  Consider the set $C = \{ i \in A : A \in \Sr, z_i^r > 0 \}$, which contains $n$ distinct vertices since the sets in $\Sr$ are pairwise disjoint.
  For $i, j \in C$, since all components of $z_i$ and $z_j$ are nonnegative, we have:
  \begin{equation*}
    X_{i,j} = \langle z_i, z_j \rangle \geq z_i^r \, z_j^r > 0
  \end{equation*}
  By constraint~\eqref{eq:opt2_con2}, this implies $(i, j) \in E_c$.
  Hence, the vertices of $C$ are pairwise adjacent and form a clique of order $n$.
\end{proof}

Combining Theorem~\ref{thm:cp_exact} with Theorem~\ref{thm:levi}, we obtain an exact conic formulation of graph isomorphism.
\begin{corollary}
  Let $G_1$ and $G_2$ be two graphs of order $n$.
  Then $G_1 \cong G_2$ if and only if problem $\mathrm{P}(\CP)$ is feasible.
  \label{cor:gi_cp}
\end{corollary}

The proof of the ``if'' direction of Theorem~\ref{thm:cp_exact} uses only constraints~\eqref{eq:opt2_con2}, \eqref{eq:opt2_con4} and complete positivity.
The remaining constraints are valid for every solution constructed in Theorem~\ref{thm:feasible} and are kept because they tighten the relaxations of Section~\ref{sec:sdp}.
Note also that Theorem~\ref{thm:cp_exact} concerns the \emph{existence} of a completely positive solution.
It does not assert that every feasible solution of the relaxations below is completely positive when $G_1 \cong G_2$.

Corollary~\ref{cor:gi_cp} does not give an efficient algorithm for graph isomorphism.
Deciding whether a given matrix is completely positive is NP-hard \cite{dickinson2014computational}, and completely positive programs are known to express NP-hard problems exactly, such as computing the stability number of a graph \cite{deklerk2002approximation} or optimizing binary quadratic programs \cite{burer2009copositive}.
The formulation is, however, a convenient starting point for convex relaxations, which we study next.

\section{The Semidefinite Relaxation}
\label{sec:sdp}
Replacing the completely positive cone $\CP$ by the larger cones $\DNN$ or $\PSD$ yields the semidefinite programs $\mathrm{P}(\DNN)$ and $\mathrm{P}(\PSD)$.
Since $\CP \subseteq \DNN \subseteq \PSD$, every feasible solution of $\mathrm{P}(\CP)$ is feasible for $\mathrm{P}(\DNN)$ and $\mathrm{P}(\PSD)$, and Corollary~\ref{cor:gi_cp} gives the following.

\begin{proposition}
  If $\mathrm{P}(\PSD)$ or $\mathrm{P}(\DNN)$ is infeasible, then $G_1$ and $G_2$ are not isomorphic.
  \label{prop:sound}
\end{proposition}

The converse does not necessarily hold, \ie a feasible solution of $\mathrm{P}(\PSD)$ or $\mathrm{P}(\DNN)$ does not, in general, certify that $G_1 \cong G_2$.
We next establish some structural properties of the relaxation.

\begin{proposition}
  Let $X \in \PSD$ satisfy constraint~\eqref{eq:opt2_con4}.
  Then $X(\one_A - \one_B) = 0$ for all $A, B \in \Sr$.
  The analogous statement holds for constraint~\eqref{eq:opt2_con5} and $\Sc$.
  Consequently, for $n \geq 2$:
  (i) constraints~\eqref{eq:opt2_con6} and \eqref{eq:opt2_con7} are implied by constraints~\eqref{eq:opt2_con4}, \eqref{eq:opt2_con5} and $X \succeq 0$; and
  (ii) every feasible solution of $\mathrm{P}(\PSD)$ is singular, so $\mathrm{P}(\PSD)$ has no strictly feasible (positive definite) solution.
  \label{prop:kernel}
\end{proposition}
\begin{proof}
  By constraint~\eqref{eq:opt2_con4} applied to the pairs $(A,A)$, $(A,B)$, $(B,A)$ and $(B,B)$, we have:
  \begin{equation*}
    \one_A^\top X \one_A - \one_A^\top X \one_B - \one_B^\top X \one_A + \one_B^\top X \one_B = (\one_A - \one_B)^\top X (\one_A - \one_B) = 1 - 1 - 1 + 1 = 0
  \end{equation*}
  Since $X \succeq 0$, a vector $y$ with $y^\top X y = 0$ satisfies $X y = 0$.
  Hence, $X(\one_A - \one_B) = 0$, which, when written entrywise, is exactly constraint~\eqref{eq:opt2_con6}.
  The same argument applies to $\Sc$ and gives constraint~\eqref{eq:opt2_con7}.
  For $n \geq 2$ and $A \neq B$, the vector $\one_A - \one_B$ is nonzero, so $X$ is singular.
\end{proof}

Proposition~\ref{prop:kernel} has a practical consequence.
Interior-point methods compute $\epsilon$-approximate solutions of semidefinite programs in time polynomial in the input size and $\log(1/\epsilon)$ under regularity assumptions such as the existence of a strictly feasible point \cite{nesterov1994interior,vandenberghe1996semidefinite}, which fails here.
This particular source of degeneracy can be removed by restricting $X$ to the orthogonal complement of the subspace spanned by the vectors $\one_A - \one_B$ (a facial reduction step), but we do not pursue this here.
In practice, $\mathrm{P}(\PSD)$ therefore provides an approximate, one-sided test for non-isomorphism.
The relaxation has $m = n^2$ rows and columns and $O(n^4)$ linear constraints.

We next compare the relaxation with the necessary condition of Section~\ref{sec:theta}.
\begin{proposition}
  If $\mathrm{P}(\PSD)$ is feasible, then $\vartheta(\bar{G}_c) = n$.
  Hence, whenever the test of Section~\ref{sec:theta} certifies that $G_1 \not\cong G_2$, so does $\mathrm{P}(\PSD)$.
  \label{prop:theta}
\end{proposition}
\begin{proof}
  Let $X$ be feasible for $\mathrm{P}(\PSD)$ and let $Y = X/n$.
  By constraints~\eqref{eq:opt2_con1}, \eqref{eq:opt2_con2} and \eqref{eq:opt2_con8}, $Y$ is feasible for the program~\eqref{eq:theta} with $G = G_c$.
  Since $\Sr$ is a partition of $V_c$, constraint~\eqref{eq:opt2_con4} gives:
  \begin{equation*}
    \langle J, Y \rangle = \frac{1}{n} \sum_{A \in \Sr} \; \sum_{B \in \Sr} \; \sum_{i \in A} \; \sum_{j \in B} X_{i,j} = \frac{1}{n} \, n^2 = n
  \end{equation*}
  Hence, $\vartheta(\bar{G}_c) \geq n$, and since $\vartheta(\bar{G}_c) \leq \chi(G_c) \leq n$, equality holds.
\end{proof}

Note that the total sum $\langle J, X \rangle$ of the entries of $X$ equals $n^2$ for every feasible solution of $\mathrm{P}(\PSD)$.
There is therefore no informative objective value to compare against.

Finally, we relate the relaxation to the Lasserre hierarchy \cite{lasserre2001global}.
Consider the following integer formulation of the existence of a clique of order $n$ in $G_c$, with a variable $x_i \in \{0,1\}$ for each $i \in V_c$:
\begin{equation}
  \sum_{i \in A} x_i = 1, \; A \in \Sr \cup \Sc, \qquad x_i x_j = 0, \; i \neq j, \; (i,j) \notin E_c
  \label{eq:ip}
\end{equation}
By Lemmas~\ref{lemma:lem1} and~\ref{lemma:lem2}, the solutions of \eqref{eq:ip} are exactly the indicator vectors of the cliques of order $n$ of $G_c$, \ie the permutation matrices that correspond to isomorphisms from $G_1$ to $G_2$.
The degree-two Lasserre relaxation of \eqref{eq:ip} asks for $y \in \mathbb{R}^m$ and $Y \in \mathbb{S}^m$ such that:
\begin{equation*}
  M = \begin{bmatrix} 1 & y^\top \\ y & Y \end{bmatrix} \succeq 0, \quad
  Y_{i,i} = y_i, \quad
  Y_{i,j} = 0 \; (i \neq j, \; (i,j) \notin E_c), \quad
  \sum_{i \in A} y_i = 1, \quad
  \sum_{i \in A} Y_{i,j} = y_j \; (j \in V_c)
\end{equation*}
for all $A \in \Sr \cup \Sc$, where the last constraint is obtained by multiplying the linear constraints of \eqref{eq:ip} by $x_j$.

\begin{proposition}
  If $(y, Y)$ is feasible for the degree-two Lasserre relaxation of \eqref{eq:ip}, then $Y$ is feasible for $\mathrm{P}(\PSD)$.
  \label{prop:lasserre}
\end{proposition}
\begin{proof}
  Constraint~\eqref{eq:opt2_con8} holds since $Y$ is a principal submatrix of $M \succeq 0$.
  Constraint~\eqref{eq:opt2_con2} is imposed directly.
  The $2 \times 2$ principal submatrix of $M$ indexed by $0$ and $i$ is $\left[\begin{smallmatrix} 1 & y_i \\ y_i & y_i \end{smallmatrix}\right] \succeq 0$, so $y_i - y_i^2 \geq 0$ and $Y_{i,i} = y_i \leq 1$, which gives \eqref{eq:opt2_con3}.
  Constraint~\eqref{eq:opt2_con1} follows from $\mathrm{trace}(Y) = \sum_{i \in V_c} y_i = \sum_{A \in \Sr} \sum_{i \in A} y_i = n$.
  For $A, B \in \Sr$, we have $\sum_{i \in A, j \in B} Y_{i,j} = \sum_{j \in B} y_j = 1$, which gives \eqref{eq:opt2_con4}, and similarly \eqref{eq:opt2_con5} for $\Sc$.
  Finally, $\sum_{i \in A} Y_{i,j} = y_j = \sum_{i \in B} Y_{i,j}$ gives \eqref{eq:opt2_con6} and \eqref{eq:opt2_con7}.
\end{proof}

Proposition~\ref{prop:lasserre} shows that $\mathrm{P}(\PSD)$ is no stronger than the degree-two Lasserre relaxation of \eqref{eq:ip}.
If $(y,Y)$ comes from a level of the hierarchy of degree at least four, then also $Y_{i,j} \geq 0$ for all $i,j$, since $x_i x_j = (x_i x_j)^2$ on $\{0,1\}^m$, and the same argument shows that $Y$ is feasible for $\mathrm{P}(\DNN)$.
There exist pairs of non-isomorphic graphs for which $\Omega(n)$ levels of the Lasserre hierarchy, applied to an integer formulation of graph isomorphism, remain feasible \cite{o2014hardness,snook2014graph}.
Provided that their formulation implies \eqref{eq:ip} at low degree, which we expect to be the case for the standard formulations, these pairs are not distinguished by $\mathrm{P}(\PSD)$ or $\mathrm{P}(\DNN)$.

\section{Conclusion and Open Questions}
In this paper, we revisited the classical reduction of graph isomorphism to the problem of deciding whether the compatibility graph of two graphs contains a clique of order $n$.
We showed that this problem can be expressed exactly as a feasibility problem over the cone of completely positive matrices (Theorem~\ref{thm:cp_exact}), so that two graphs are isomorphic if and only if this completely positive program is feasible (Corollary~\ref{cor:gi_cp}).
Since membership in the completely positive cone is NP-hard to decide, this formulation does not yield an efficient algorithm on its own.
Relaxing complete positivity to positive semidefiniteness gives a semidefinite program whose infeasibility certifies non-isomorphism.
We showed that this relaxation is at least as strong as the test based on the Lov\'asz number of the complement of the compatibility graph, that it is never strictly feasible, and that it is implied by the degree-two level of the Lasserre hierarchy applied to a natural integer formulation of graph isomorphism.
Given the frequent use of semidefinite programming in the design of algorithms \cite{goemans1995improved,srivastav1998finding,feige2000finding}, it seemed worthwhile to investigate these formulations for the graph isomorphism problem.
The results of this paper constitute a first step in this direction, and they leave several open questions.

The central question is to characterize when the relaxation is tight, \ie for which pairs of non-isomorphic graphs $G_1, G_2$ the problems $\mathrm{P}(\PSD)$ and $\mathrm{P}(\DNN)$ are infeasible, and for which classes of graphs feasibility of one of these problems implies $G_1 \cong G_2$.
Related questions include how these relaxations compare with color refinement and with the $k$-WL algorithm for small $k$, and in particular whether they distinguish pairs of non-isomorphic regular or strongly regular graphs.
Whether the dual of $\mathrm{P}(\PSD)$ yields interpretable certificates of non-isomorphism; and whether hierarchies of tractable outer approximations of the completely positive cone yield practically useful tests between $\mathrm{P}(\DNN)$ and the exact formulation $\mathrm{P}(\CP)$.

\section*{Acknowledgments}
The authors would like to thank Prof. Leo Liberti for discussions and comments on early versions of this paper.

\bibliographystyle{ieeetr}
\bibliography{biblio}

\end{document}